\documentclass[conference]{IEEEtran}
\usepackage{algorithm, algorithmic,amsthm}

\newtheorem{theorem}{Theorem}

\newtheorem{lemma}{Lemma}
\usepackage{tabularx}
\usepackage{array}

\usepackage{setspace}
\usepackage{anyfontsize}
\renewcommand{\IEEEbibitemsep}{0pt plus 0.5pt}
\makeatletter
\IEEEtriggercmd{\reset@font\normalfont\fontsize{7.9pt}{8.40pt}\selectfont}
\makeatother
\IEEEtriggeratref{1}

\usepackage{ifpdf}
\ifCLASSINFOpdf
\else
\fi
\usepackage{url}
\usepackage{amsmath,graphicx}
\usepackage{placeins}
\usepackage{amsmath,epsfig}
\usepackage{epstopdf}
\usepackage{amssymb}

\usepackage{dblfloatfix}
\usepackage{amssymb}
\usepackage{caption}
\usepackage{comment}
\usepackage{subfigure}
\usepackage{pstricks}
\usepackage{subfloat}
\begin{document}
%\fontsize{9}{11}\selectfont
%
% paper title
% can use linebreaks \\ within to get better formatting as desired
\title{\LARGE A Low-Complexity Detection Algorithm For Uplink Massive MIMO Systems  Based on Alternating Minimization\vspace{-.3in}} 
\vspace{-2in}%}
% author names and affiliations
%\lhead{Submitted to IEEE Wireless Communication Letters}
% use a multiple column layout for up to three different
%2-Stage LASSO ADMM For Large Scale MIMO Detection}
% affiliations
\author{\IEEEauthorblockN{}
\IEEEauthorblockA{\vspace{-1in}Anis Elgabli, Ali Elghariani, Vaneet Aggarwal, and Mark R. Bell%\\%\\
%Purdue University, West Lafayette IN 47907, *University of Tripoli, Libya\\
%Email: aelgabli, vaneet, mrb@purdue.edu and a.elghariani@uot.edu.ly
}}
\maketitle
\vspace{-.1in}
\begin{abstract}
\if0

This paper explores the benefit of using some of the machine learning techniques and Big data optimization tools in approximating maximum likelihood (ML) detection of Large Scale MIMO systems. First, large scale MIMO detection problem is formulated as a LASSO (Least Absolute Shrinkage and Selection Operator) optimization problem. Then, Alternating Direction Method of Multipliers (ADMM) is considered in solving this problem. The choice of ADMM is motivated by its ability of solving convex optimization problems by breaking them into smaller sub-problems, each of which are then easier to handle. Further improvement is obtained using two stages of LASSO with interference cancellation from the first stage. The proposed algorithm is investigated at various modulation techniques with different number of antennas. It is also compared with widely used algorithms in this field. Simulation results demonstrate the efficacy of the proposed algorithm for both uncoded and coded cases.   
\fi
In this paper, we propose an algorithm based on the Alternating Minimization technique to solve the uplink massive MIMO detection problem. \textcolor{black}{The proposed algorithm is specifically designed to avoid any matrix inversion and any computations of the Gram matrix at the receiver. The algorithm provides a lower complexity compared to the conventional MMSE detection technique, especially when the total number of user equipment (UE) antennas (across all users) } is close to the number of base station (BS) antennas. The idea is that the algorithm re-formulates the maximum likelihood (ML) detection problem as a sum of convex functions based on decomposing the received vector into multiple vectors. Each vector represents the contribution of one of the transmitted symbols in the received vector. Alternating Minimization is used to solve the new formulated problem in an iterative manner with a closed form solution update in every iteration. Simulation results demonstrate the efficacy of the proposed algorithm in the uplink massive MIMO setting for both coded and uncoded cases.

\end{abstract}


\section{Introduction}
\label{sec:intro}

Massive  multiple-input  multiple-output  (MIMO)  is  one of  the  most  promising  techniques  for  the 5th  Generation  (5G) networks due to its potential for enhancing throughput, spectra efficiency, and energy efficiency \cite{boccardi2014five}, \cite{van2017massive}. Massive MIMO requires the BS to be equipped with arrays of hundreds of antennas to serve tens of user terminals with single or multiple antennas.

%The main idea is to equip the base station (BS) with hundreds of antennas that serve a relatively small number of users in the orders of tens simultaneously and in the same frequency band. 

%Theoretical results for massive MIMO not only promise higher peak data rates, improved coverage, and longer range, but also that simple, low-complexity, and energy-efficient detection and precoding algorithms are able to achieve optimum performance in the large-antenna limit, i.e., where the number of BS antennas approaches infinity [1]?[4].

Theoretical results of massive MIMO show that linear detectors such as zero forcing (ZF) and minimum mean square error (MMSE) can achieve optimum performance under the favorable propagation conditions \cite{lu2014overview}. The favorable propagation condition means that the number of BS antennas grows very large compared to the number of UE antennas, which leads to the column-vectors of the propagation matrix to be asymptotically orthogonal. As a result of this orthogonality, the ZF and MMSE detectors can be implemented with simple diagonal inversions \cite{rusek2013scaling}. 

The current practical number of the BS antennas in massive MIMO systems is in the order of tens to a hundred. This is far from the theoretical limit that leads to the orthogonality mentioned above \cite{rusek2013scaling}. Therefore, the linear detectors still need to perform a matrix inversion for the signal detection of the uplink massive MIMO system, which entails extensive computational complexity \cite{zhang2018low}. Neumann series expansion, Cholesky decomposition, and successive over-relaxation techniques are proposed in the literature to reduce the complexity of the matrix inversion process in the MMSE detector \cite{gao2014matrix} \cite{yin2013implementation}. These approaches require a lower computational complexity than the exact matrix inversion while delivering near-optimal results only for theoretical massive MIMO configurations, \textcolor{black}{that is, when the ratio between the number of BS antennas and the number of single antenna users is large enough (e.g., $\ge 16$) \cite{wu2014large}. In realistic massive MIMO scenarios where this ratio is small, \cite{wu2014large} indicates that a large truncation order is required for Neumann series expansion technique which makes the computational complexity higher than the exact matrix inversion operations. Moreover, implementing these approximations significantly deteriorate performance as compared to the exact MMSE performance \cite{gao2014matrix}}. 

In this letter we present a novel formulation of a low complexity iterative algorithm based on Alternating Minimization, referred to as AltMin. This algorithm provides similar bit error rate (BER) performance to the exact matrix inverted MMSE technique, with one order less complexity.% compared to MMSE, with no matrix inversion. %It is also possible to have even better performance than the MMSE at the cost of increasing number of iterations. It

%It is important to note that, unlike previous work done in this problem, our proposed algorithm does not start from the MMSE problem and try to reduce inversion process, it is basically propose a new approximate formulation to the ML. Also the proposed algorithm approaches the MMSE performance even at large number of users (unlike previous work [ref]

The proposed algorithm approximates ML detection problem as a sum of convex functions based on decomposing the received vector into multiple vectors. Each vector represents the contribution of one of the transmitted symbols in the received vector. Then, Alternating Minimization is used to solve the new formulated problem in an iterative manner with a closed form solution update in every iteration that does not require any matrix inversion or any matrix multiplications. 

\textcolor{black}{Although there are several algorithms presented in the literature for the uplink massive MIMO detection problem, the key contributions include the following: (1) the re-formulation of the problem is novel, (2) the alternating minimization used to solve the proposed formulation has a closed form expression update at every iteration. The algorithm avoids any matrix inversion and any computations of the Gram matrix (i.e., $\textbf{H}^{H} \textbf{H}$). Thus, the proposed approach is low-complexity and has been shown empirically to perform better than the considered baseline. }
\vspace{-.1in}
\section{System Model and Problem Formulation}
\vspace{-.1in}
\subsection{System Model}
Consider the uplink data detection in a multi-user (MU) massive MIMO system with $N_r$ BS antennas and $N_t$ UE antennas. The vector $\tilde{\textbf{x} }= (\tilde{x}_1, \tilde{x}_2, \dots, \tilde{x}_{N_t})^T \in  \mathbb{C}^{N_t  \times1}$ represents the complex transmitted signal, where $x_k$ is the transmitted symbol for user $k$ with $E{|\tilde{x}_i|^ 2} = 1 , \forall i$. Each user transmits symbols over a flat fading channels and the signals are demodulated and sampled at the receiver. The vector $\tilde{\textbf{y}} = (\tilde{y}_1, \tilde{y}_2, \dots , \tilde{y}_{N_r} )^T \in \mathbb{C}^{{N_r} \times1}$ represents the complex received signal, and the channel matrix $\tilde{\textbf{H}}  \in \mathbb{C}^{{N_r} \times N_t} $ can be represented as $(\tilde{\textbf{h}}_1, \tilde{\textbf{h}}_2, \dots, \tilde{\textbf{h}}_{N_t})$, where $\tilde{\textbf{h}}_i =(\tilde{h}_{1,i}, \tilde{h}_{2,i},..., \tilde{h}_{N_r,i})^T  \in \mathbb{C}^{{N_r} \times 1}$, and \textcolor{black}{$\tilde{h}_{m,n}$ is the fading channel gain from transmit antenna $n$ to the receive antenna $m$ that is assumed to be i.i.d. (independent and identically distributed) complex Gaussian random variables with zero mean and unit variance, i.e., $\tilde{h}_{m,n}  \sim \mathcal{CN}(0,1)$}. The system can be modeled as; $\tilde{\textbf{{y}}}=\tilde{\textbf{{H}}} \tilde{\textbf{{x}}}+\tilde{\textbf{{v}}}$,
%In the Rayleigh flat fading channel model, all channel gains are assumed to be independent and identically distributed complex Gaussian random variables with zero mean and unit variance.
%\begin{equation}
%\tilde{\textbf{{y}}}=\tilde{\textbf{{H}}} \tilde{\textbf{{x}}}+\tilde{\textbf{{v}}}
%\label{eq: 1}
%\end{equation}
where the $\tilde{\textbf{v}} = (\tilde{v}_1, \tilde{v}_2, иии , \tilde{v}_{N_r} )^T \in \mathbb{C}^{{N_r} \times N_t} $ is the complex additive white Gaussian noise (AWGN) vector whose elements are mutually independent \textcolor{black}{with zero mean and variance $\sigma^2_v$}. %We assume that the transmitted signals are i.i.d Gaussian distribution and satisfy $E{|x_i|^ 2} = 1 , \forall i$, i.e. The symbols are normalized to unit energy.
\if0
\begin{figure}[htp]
%%%\begin{minipage}[b]{1.0\linewidth}
 \centering
\includegraphics[width=6cm, height=4cm]{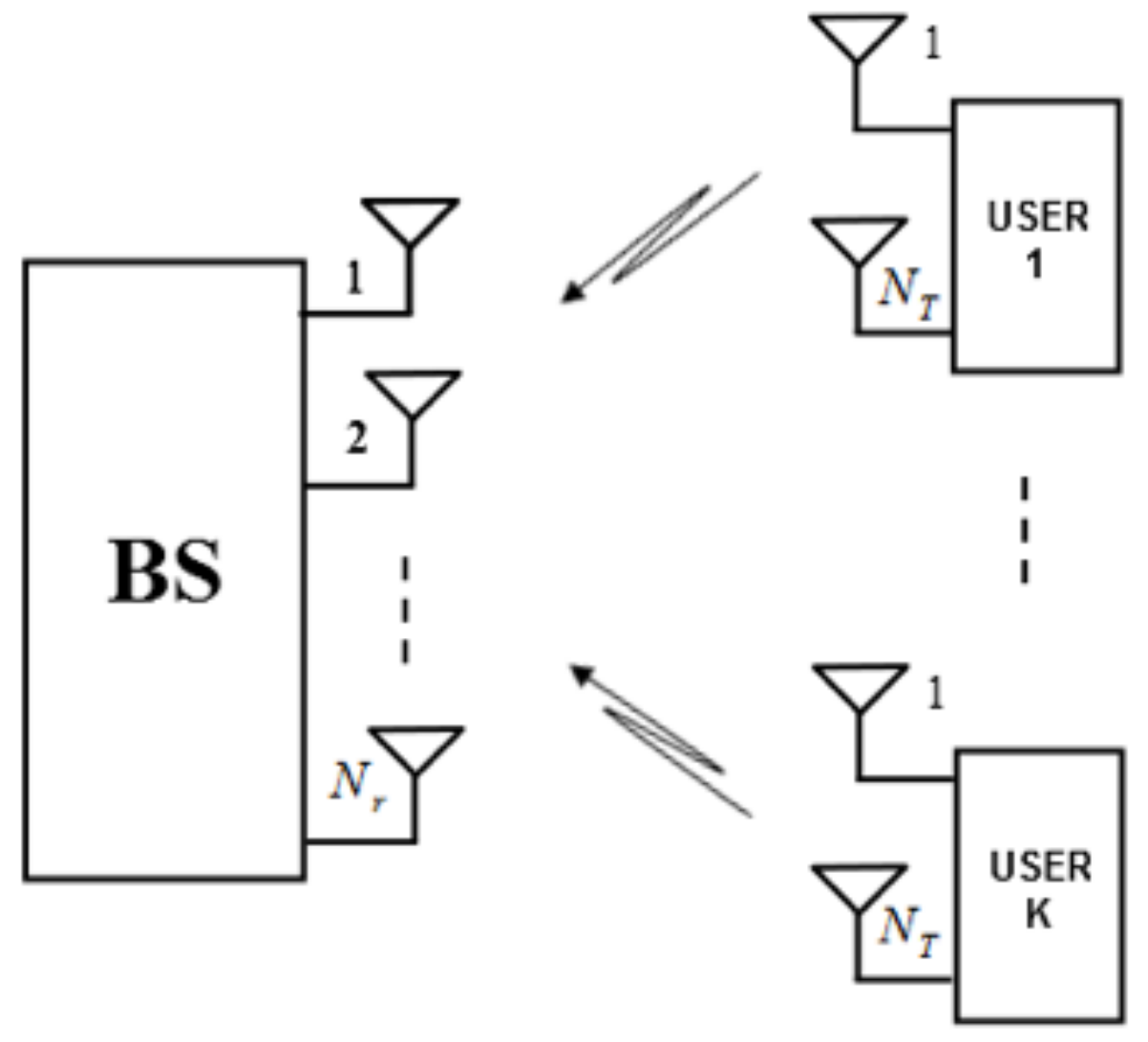}
\caption{\small  System Model for Massive MIMO uplink}
\label{fig:1}
\end{figure}
\fi
The corresponding real-valued system model is $\textbf{{y}}=\textbf{{H}}\textbf{{x}}+\textbf{{v}}$ \cite{elgabli2017two}, \cite{elghariani2016low}, \textcolor{black}{where $\textbf{{x}} \in  \mathbb{R}^{2N_t  \times1}$, $\textbf{{y}} \in  \mathbb{R}^{2N_r  \times1}$, $\textbf{{v}} \in  \mathbb{R}^{2N_r  \times1}$, and $\textbf{{H}} \in  \mathbb{R}^{2N_r  \times 2N_t}$}. The equivalent ML detection problem of the real model can be written in the form $\widehat{{\textbf{x}}}=  \underset{{{\textbf{x}}\in\chi^{2N_t}}}{\text{argmin}} \parallel{{\textbf{y}}}-{{\textbf{H}}}{{\textbf{x}}}\parallel_2^{2}$, where $ \chi=\frac{1}{\Gamma}\{-\sqrt{\textit{M}}+1,..,-1,1,...,\sqrt{\textit{M}}-1\}$, $\textit{M} \normalsize$ is the constellation size, and $\frac{1}{\Gamma}$ is for normalization factor.

%%%%%%%%%%%%%%%%%%%%%%%%%%%%%%%%%%%%%%%%%%%%%%%%%%%%%%%%%%%%%%%%%%%%%%%%%%%%%
%$\parallel{{\textbf{y}}}-{{\textbf{H}}}{{\textbf{x}}}\parallel_2^{2} =  \parallel{({{\textbf{y}_1}}-{{\textbf{h}_1}}{{{x}_1}})  + ({{\textbf{y}_2}}-{{\textbf{h}_2}}{{{x}_2}})  + \dots +  ({{\textbf{y}_{2N_t}}}-{{\textbf{h}_{2N_t}}}{{{x}_{2N_t}}}) }\parallel_2^{2}$

%$\parallel{{\textbf{y}}}-{{\textbf{H}}}{{\textbf{x}}}\parallel_2^{2} ~\leq  ~\sum_{i=1}^{2N_{t}} \parallel{{\textbf{y}_i}}-{{\textbf{h}_i}}{{{x}_i}}\parallel_2^{2}$

\if0
\small
\begin{equation}
\widehat{{\textbf{x}}}=  \underset{{{\textbf{x}}\in\chi^{2N_t}}}{\text{argmin}} \parallel{{\textbf{y}}}-{{\textbf{H}}}{{\textbf{x}}}\parallel_2^{2}
\label{eq: 6}
\end{equation}
\normalsize
\fi

\vspace{-.1cm}
\subsection{Problem Formulation}
First, we decompose the received vector $\textbf{y} $ into a linear combination of vectors so that $\textbf{y} = \sum_{i=1}^{2N_{t}}{\textbf{y}_i}$, where $\textbf{y}_i$ represents the contribution of the $i$-th transmitted symbol in the received vector. The element wise representation of the decomposed received vector is: 
\small
\\
\\
$
\begin{bmatrix}
y^{(1)}\\
.\\
.\\
y^{(k)}\\
.\\
.\\
y^{(2N_r)}\\
\end{bmatrix}
=
\begin{bmatrix}
y_1^{(1)}\\
.\\
.\\
y_1^{(k)}\\
.\\
.\\
y_1^{(2N_r)}\\
\end{bmatrix}
+\cdots+
\begin{bmatrix}
y_i^{(1)}\\
.\\
.\\
y_i^{(k)}\\
.\\
.\\
y_i^{(2N_r)}\\
\end{bmatrix}+\cdots+
\begin{bmatrix}
y_{2N_{t}}^{(1)}\\
.\\
.\\
y_{2N_{t}}^{(k)}\\
.\\
.\\
y_{2N_{t}}^{(2N_r)}\\
\end{bmatrix}
$
\\
\\
\normalsize
\textcolor{black}{where $y_i^{(k)}$ represents the $k$-th element of the decompose vector $\textbf{y}_i$}. Thus, the  $k$-th element of the real-valued received vector $\textbf{y}$ can be represented as ${y}^{(k)} = \sum_{i=1}^{2N_{t}}{{y}_i^{(k)}}$, $k=1,\dots, 2N_r$. Let \textcolor{black}{$\textbf{h}_i$ be the $i$-th column of the real-valued channel matrix $\textbf{H}$}.  Now, we relax the non-convexity constraint on the feasible set $\chi$, and approximate the ML problem based on the above decomposition as follows: 
\begin{equation}
\underset{{x_i},\textbf{y}_i \forall i}{\text{argmin}} \sum_{i=1}^{2N_{t}} \parallel{{\textbf{y}_i}}-{{\textbf{h}_i}}{{{x}_i}}\parallel_2^{2} \,\,\,\,\, \text{subject to}
\label{eq: 7}
\end{equation}
%\,\,\,\,\,\,\,\,\,\,\, subject to
\begin{equation}
\sum_{i=1}^{2N_{t}}{{y}_i}^{(k)}={y}^{(k)}, \forall k=1,\cdots,2N_r
\label{eq: 7c1}
\end{equation}
\begin{equation}
-l \leq {x}_i \leq l, \forall i=1,\cdots,2N_t
\label{eq: 7c2}
\end{equation}
Where $l= \frac{1}{\Gamma} (\sqrt{\textit{M}}-1)$. The objective function in (\ref{eq: 7}) is a sum of separable terms, each of which is a function of only one symbol and its contribution in the received vector. In the next section, we use AltMin to solve the proposed formulation.

 \if0
 \subsection{K.K.T Conditions}
 
 The problem in~\eqref{eq: 7}-\eqref{eq: 7c2} is convex optimization problem with the linear constraints, so K.K.T. conditions are sufficient and necessary for the optimal solution. In this section, we use the K.K.T. conditions
to characterize the optimality of the problem in~\eqref{eq: 7}-\eqref{eq: 7c2}.
 
 The Lagrangian function for~\eqref{eq: 7}-\eqref{eq: 7c2} is defined as follows:
 
 \begin{align}
&\mathcal{L}=\sum_{i=1}^{2N_{t}} \parallel{{\textbf{y}_i}}-{{\textbf{H}_i}}{{\textbf{x}_i}}\parallel_2^{2}+\sum_{k=1}^{2N_r}\lambda^{k}(\textbf{y}^{(k)}-\sum_{i=1}^{2N_{t}}{\textbf{y}_i}^{(k)})\nonumber\\&+\mu_1(l-\textbf{x}_i)+\mu_2(l+\textbf{x}_i)
\end{align}

Then, the following K.K.T. conditions, which are sufficient and necessary for the optimal solution to the convex optimization problem in~\eqref{eq: 7}-\eqref{eq: 7c2}, are obtained from the Lagrangian function:

\begin{equation}
\textbf{y}_i^{(k)}=\textbf{H}_i^{(k)}\textbf{x}_i+\lambda^{(k)}/2, \forall i,k
\label{eq: 10}
\end{equation}

\begin{equation}
\lambda^{(k)}=\frac{1}{N_t}(\textbf{y}^{(k)}-\sum_{i=1}^{2N_{t}}\textbf{H}_i^{(k)}\textbf{x}_i), \forall i,k
\label{eq: 11}
\end{equation}

\begin{equation}
2\textbf{x}_i\sum_{k=1}^{2N_t}{\textbf{H}_i^{(k)}}^2-2\sum_{k=1}^{2N_t}\textbf{y}_i^{(k)}\textbf{H}_i^{(k)}-\mu_1+\mu_2=0, \forall i,k
\label{eq: 131}
\end{equation}

\begin{equation}
\mu_1^{(i)}(l-\textbf{x}_i)=0, \forall i
\label{eq: 141}
\end{equation}

\begin{equation}
\mu_2^{(i)}(l+\textbf{x}_i)=0, \forall i
\label{eq: 151}
\end{equation}
\fi

\vspace{-.1in}
\section{Proposed Algorithm}

%\subsection{Alternating minimization}
The optimization problem~\eqref{eq: 7}-\eqref{eq: 7c2} is strictly and jointly convex with respect to \textbf{x} and ${\bf{\cal Y}}$ where, ${\bf{\cal Y}}=\big\{\textbf{y}_i \forall i: i\in\{1,2,\cdots,2N_t\}\big\}$. Moreover, there is no common constraint that combines both \textbf{x}  and ${\bf{\cal Y}}$. Therefore, in order to efficiently solve this problem, we first decompose it into the following two subproblems:
\begin{itemize}
\item Given $\textbf{x}$, we obtain ${\bf{\cal Y}}$ by solving %the following subproblem:
\begin{equation}
\underset{\textbf{y}_i}{\text{argmin}}\sum_{i=1}^{2N_{t}} \parallel{{\textbf{y}_i}}-{{\textbf{h}_i}}{x_i}\parallel_2^{2}\quad \text{subject to} \, \eqref{eq: 7c1}
\label{eq: 8}
\end{equation}
\item Given ${\bf{\cal Y}}$, we obtain \textbf{x} by solving %the following subproblem:
\begin{equation}
\underset{x_i}{\text{argmin}}\sum_{i=1}^{2N_{t}} \parallel{{\textbf{y}_i}}-{{\textbf{h}_i}}{{x_i}}\parallel_2^{2}\quad \text{subject to} \, \eqref{eq: 7c2}
\label{eq: 9}
\end{equation}
%\textcolor{red}{Then, we propose, AltMin (Algorithm 1), an iterative algorithm that alternatively solves~\eqref{eq: 8} for ${\bf{\cal Y}}$ and~\eqref{eq: 9} for $\textbf{x}$.}
%\textcolor{red}{Then, we propose, AltMin, an iterative algorithm that alternatively solves~\eqref{eq: 8} for ${\bf{\cal Y}}$ and~\eqref{eq: 9} for $\textbf{x}$.}

\end{itemize}
Then, we propose, AltMin, an iterative algorithm that alternatively solves~\eqref{eq: 8} for ${\bf{\cal Y}}$ and~\eqref{eq: 9} for $\textbf{x}$ given the other. Note that the respective \textcolor{black}{Karush-Kuhn-Tucker  (K.K.T)} conditions \cite{kuhn1951} of the above two subproblems form the complete set of the K.K.T. conditions for the original problem.

\vspace{-.1in}
\subsection{{\bf Solving the subproblem \eqref{eq: 8} }}

The Lagrangian function of~\eqref{eq: 8} can be written as:
\begin{equation}
\mathcal{L}=\sum_{i=1}^{2N_{t}} \parallel{{\textbf{y}_i}}-{{\textbf{h}_i}}{x_i}\parallel_2^{2}+\sum_{k=1}^{2N_r}\lambda^{k}(y^{(k)}-\sum_{i=1}^{2N_{t}}{y_i}^{(k)})
\label{eq: 100}
\end{equation}
Therefore, by solving the above Lagrangian function, which is a function of $\lambda^{(k)}$  and $\textbf{y}_i$, we get the following closed form expression updates for every element of $\lambda^{(k)}$ and $y_i^{(k)}$:
\begin{equation}
\lambda^{(k)}=C \cdot\frac{1}{N_t}(y^{(k)}-\sum_{i=1}^{2N_{t}}h_i^{(k)}x_i), \forall k
\label{eq: 11}
\end{equation}

\begin{equation}
y_i^{(k)}=h_i^{(k)}x_i+\lambda^{(k)}/2, \forall i,k
\label{eq: 10}
\end{equation}
%Therefore, the update of the $k$-th element in the vector $\textbf{y}_i$ reduces to solving~\eqref{eq: 11} in order to find $\lambda^{(k)}$ and then substituting in~\eqref{eq: 10} to update $y_i^{(k)}$. 
\textcolor{black}{where $h_i^{(k)}$ represents the k-th element in column vector $\textbf{h}_i$, which corresponds to the $i^{\text{th}}$ column of the the real-valued channel matrix $\textbf{H}$}. Note that in the update of $\lambda^{(k)}$, we introduce the scaling factor, $C$. We show in Section III-C that the proposed algorithm is optimal when $C=1$. However, we empirically find that when $C=N_t$, the number of iterations for convergence drops significantly, see Section IV-A. 
\vspace{-.1cm}
\subsection{{\bf Solving the subproblem~\eqref{eq: 9} }}

The objective function~\eqref{eq: 9} is separable with respect to every element $x_i$ in the vector \textbf{x}, and no constraint combines the elements of \textbf{x}. Therefore, the update of the $i$-th element in the vector \textbf{x} reduces to solving the following subproblem.

\begin{equation}
\underset{x_i}{\text{argmin}} \sum_{k=1}^{2N_r}(y_i^{(k)}-h_i^{(k)}x_i)^{2} \quad \text{subject to }~\eqref{eq: 7c2}%\parallel\textbf{y}_i-\textbf{h}_ix_i\parallel_2^2
\label{eq: 12}
\end{equation}
The corresponding Lagrangian function of ~\eqref{eq: 12} is:
\begin{equation}
\mathcal{L}=\sum_{k=1}^{2N_r}(y_i^{(k)}-h_i^{(k)}x_i)^{2}+\mu_1^{(i)}(l-x_i)+\mu_2^{(i)}(l+x_i)
\end{equation}
Then, the following K.K.T. conditions \cite{kuhn1951} which are sufficient and necessary for the optimal solution to the convex optimization problem in \eqref{eq: 12} are:% obtained from the Lagrangian function:

\begin{equation}
2x_i\sum_{k=1}^{2N_r}{h_i^{(k)}}^2-2\sum_{k=1}^{2N_r}y_i^{(k)}h_i^{(k)}-\mu_1^{(i)}+\mu_2^{(i)}=0
\label{eq: 13}
\end{equation}
\begin{equation}
\mu_1^{(i)}(l-x_i)=0, \,\,\mu_2^{(i)}(l+x_i)=0, \,\, \mu_1^{(i)},\mu_2^{(i)}\geq 0
\label{eq: 14}
\end{equation}
\if0
\begin{equation}
\mu_1^{(i)}(l-x_i)=0
\label{eq: 14}
\end{equation}

\begin{equation}
\mu_2^{(i)}(l+x_i)=0
\label{eq: 15}
\end{equation}

\begin{equation}
\mu_1^{(i)},\mu_2^{(i)}\geq 0
\label{eq: 16}
\end{equation}
\fi
In order to solve~\eqref{eq: 13}-\eqref{eq: 14} for every element in the vector \textbf{x}, among the following $\{\mu_1^{(i)}, \mu_2^{(i)}, x_i\}$ choices, we choose the one that minimizes~\eqref{eq: 12}: 
\vspace{-.1cm}
  \begin{equation} \mu_1^{(i)}=0,\text{ and }\mu_2^{(i)}=0 \rightarrow x_i=\frac{\sum_{k=1}^{2N_r}y_i^{(k)}h_i^{(k)}}{\sum_{k=1}^{2N_r}{h_i^{(k)}}^2}
  \label{eq:x1Update}
\end{equation}  

  \begin{equation} \mu_1^{(i)}=0,\text{ and }\mu_2^{(i)} \neq 0\rightarrow x_i=-l \label{eq:x1Update2}\end{equation}

 \begin{equation} \mu_1^{(i)}\neq 0,\text{ and }\mu_2^{(i)}= 0\rightarrow x_i=l \label{eq:x1Update3}\end{equation}
Note, we exclude the choice  $\mu_1^{(i)}\neq 0$, and $\mu_2^{(i)} \neq 0$ since $x_i$ cannot be equal to $-l$ and $l$ at the same time. %Finally, we can clearly see that solving~\eqref{eq: 13}-\eqref{eq: 16} yields a closed form expression update for the vector $\textbf{x}$. 

%To obtain the optimal solution to the proposed optimization problem~\eqref{eq: 7}-\eqref{eq: 7c2}, we propose, AltMin, an iterative algorithm that alternatively solves~\eqref{eq: 8} for ${\bf{\cal Y}}$ and~\eqref{eq: 9} for $\textbf{x}$. To perform the algorithm, we initially set $\textbf{x}$ to $\textbf{0}$, and solve~\eqref{eq: 8} to obtain the initial ${\bf{\cal Y}}$; with the updated ${\bf{\cal Y}}$, we then solve ~\eqref{eq: 9} to update $\textbf{x}$.
%The proposed iterative algorithm is summarized in Algorithm 1. In the next subsection, we will show that Algorithm 1 converges and the pairwise optimal $\textbf{x}$ and ${\bf{\cal Y}}$ can be obtained, which is also the optimal solution to the problem in ~\eqref{eq: 7}-\eqref{eq: 7c2}.

To obtain the optimal solution to the proposed optimization problem~\eqref{eq: 7}-\eqref{eq: 7c2}, AltMin solves~\eqref{eq: 8} for ${\bf{\cal Y}}$ and~\eqref{eq: 9} for $\textbf{x}$. To perform the algorithm, we initially set $\textbf{x}$ to $\textbf{0}$, and solve~\eqref{eq: 8} to obtain the initial ${\bf{\cal Y}}$. With updated ${\bf{\cal Y}}$, we  solve ~\eqref{eq: 9} to update $\textbf{x}$. The proposed algorithm is summarized in Algorithm \ref{alg1}. 
%In the next subsection, we will show that Algorithm 1 converges and the pairwise optimal $\textbf{x}$ and ${\bf{\cal Y}}$ can be obtained, which is also the optimal solution to the problem in ~\eqref{eq: 7}-\eqref{eq: 7c2}.

\begin{figure}

		%\vspace{-.6in}
		\begin{minipage}{\linewidth}
			\begin{algorithm}[H]
					{\tiny 
				\small
				\begin{algorithmic}[1]
				    \STATE {\bf 1: Initilization}
    
   \STATE $t=0$, $x_i=0 \forall i$ 
   \STATE update $\lambda^{(k)}\forall k$ according to equation~\eqref{eq: 11}
   \STATE update $y_i^{(k)}\forall i,k$ according to equation~\eqref{eq: 10}
%\STATE Solve~\eqref{eq: 8}-\eqref{eq: 8c1} to update ${\bf{\cal Y}}$
\STATE $V^{(t)}=\sum_{i=1}^{2N_{t}} \parallel{{\textbf{y}_i}}-{{\textbf{h}_i}}{{x_i}}\parallel_2^{2}$
\STATE $\delta =$ convergence tolerance%, $T=$ Maximum number of iterations
\STATE $T=$ Maximum number of iterations

    \STATE {\bf 2: Alternating Minimization:}
    \REPEAT%\STATE {\bf REPEAT}
 \STATE $t \leftarrow t+1$
    \STATE update $x_i\forall i$ according to ~\eqref{eq:x1Update}
 \STATE update $\lambda^{(k)}\forall k$ according to equation~\eqref{eq: 11}
   \STATE update $y_i^{(k)}\forall i,k$ according to equation~\eqref{eq: 10}

\STATE $V^{(t)}=\sum_{i=1}^{2N_{t}} \parallel{{\textbf{y}_i}}-{{\textbf{h}_i}}{{x_i}}\parallel_2^{2}$
\UNTIL  $|V^{(t)}-V^{(t-1)}| < \delta$ {\bf OR}  $t > T$%\STATE {\bf UNTIL } $|V^{(t)}-V^{(t-1)}| < \delta$ {\bf OR}  $t > T$

   				\end{algorithmic}
				\caption{Proposed Algorithm } \label{alg1}
}						
			\end{algorithm}
		\end{minipage}
		\vspace{-.2in}
	\end{figure}

\vspace{-.1in}
\subsection{{\bf Optimality of AltMin Algorithm in solving~\eqref{eq: 7}-\eqref{eq: 7c2}}}

%\textcolor{red}{We now show the optimality of AltMin when $C=1$}. I
In order to show the optimality of the iterative algorithm (when $C=1$), we first give the following lemma.
\begin{lemma}
Given $\textbf{x}$, the optimal ${\bf{\cal Y}}$ for the problem~\eqref{eq: 8}  is unique. Similarly, given ${\bf{\cal Y}}$, the optimal $\textbf{x}$ to~\eqref{eq: 9} is unique. \label{lem1}
%\label{lemma}
\end{lemma}

\begin{proof}
This lemma can be obtained by verifying the strict convexity of~\eqref{eq: 7}-\eqref{eq: 7c2} with respect to $\textbf{x}$ given ${\bf{\cal Y}}$, and with respect to ${\bf{\cal Y}}$ given $\textbf{x}$.  %$\blacksquare$
\end{proof}

\begin{theorem}
AltMin Algorithm converges, and (${\bf{\cal Y}}$, $\textbf{x}$) is optimal to~\eqref{eq: 7}-\eqref{eq: 7c2}. 
\label{theorem1}
\end{theorem}

\begin{proof}
\textcolor{black}{We use the result in \cite{beck2015convergence} to show the optimality and convergence of alternating minimization, which requires us to show five conditions. The first condition is satisfied since the constraints do not have both variables together and are linear functions. The second condition indicates that the objective is a continuously differentiable convex function which also holds as the Hessian matrix of the objective function is positive semidefinite. The third and the fourth conditions are satisfied since the objective function is uniformly Lipschitz continuous with respect to each of the variable. The fifth condition is that the two sub-problems have minimizers which holds by Lemma \ref{lem1}. This proves the result as in the statement of the Theorem. }
%\textcolor{blue}{Since ${\bf{\cal Y}}$, and $\textbf{x}$ are updated by successively solving~\eqref{eq: 8} and ~\eqref{eq: 9} in every iteration, the obtained objective value is non-increasing over iterations. The objective function is non-increasing because it is jointly convex with respect to both ${\bf{\cal Y}}$, and $\textbf{x}$. Moreover, the  dual variables associated with constraints on ${\bf{\cal Y}}$, and $\textbf{x}$ are disjoint. i.e, there is no a constraint that combines ${\bf{\cal Y}}$, and $\textbf{x}$.  At the convergence point, ${\bf{\cal Y}}$, and $\textbf{x}$ are jointly optimal to the subproblems, since otherwise the objective value can be further decreased in the next iteration. Then, by Lemma 2, the obtained (${\bf{\cal Y}}$, $\textbf{x}$) is the optimal solution to~\eqref{eq: 7}-\eqref{eq: 7c2}. $\blacksquare$.}
\end{proof}
%\cite{wang2016optimal}

%intuitively, the optimality can be proved by first verifying the pairwise optimality of the solution upon convergence
%and then showing it cannot be suboptimal.

\vspace{-.1in}
\subsection{{\bf Complexity Analysis of AltMin}}

The update of ${\bf{\cal Y}}$ has a complexity of $O(N_tN_r)$, while the update of $\textbf{x}$ has a complexity of $O(N_t)$. The AltMin algorithm performs $T$ iterations between these updates before converging to the optimal solution. Therefore, the overall complexity of AltMin in solving~\eqref{eq: 7}-\eqref{eq: 7c2} is  $O(T N_t N_r)$, which is notably lower than $O(N_t^3)$ of MMSE for large $N_t$, as demonstrated in Section \ref{numRes}.

%Alternating minimization suffers from slow conversion $O(1/\sqrt(t))$ where $t$ is the iteration number, so $T$ could be large. Therefore, in next section we consider ADMM which is shown to has faster convergence rate $O(1/t)$. Moreover, as we will show in the evaluation section, stopping the algorithm after few iterations and rounding solutions to the nearest feasible solution will still yield comparable performance compared to MMSE estimator and lower complexity.

\vspace{-.1in}
\section{Numerical Results}
\label{numRes}
Numerical results for coded and uncoded uplink massive MIMO systems in a block flat fading channel is presented.  We assume perfect knowledge of the channel state information at the receiver. QPSK modulation is considered for demonstration; however, the proposed algorithm can be extended for higher QAM modulations in a straightforward manner. The performance and computational complexity of the proposed algorithm is compared with the linear MMSE detector.

\begin{figure*}[!htb]
\minipage{0.32\textwidth}
  \includegraphics[width=6cm, height=4cm]{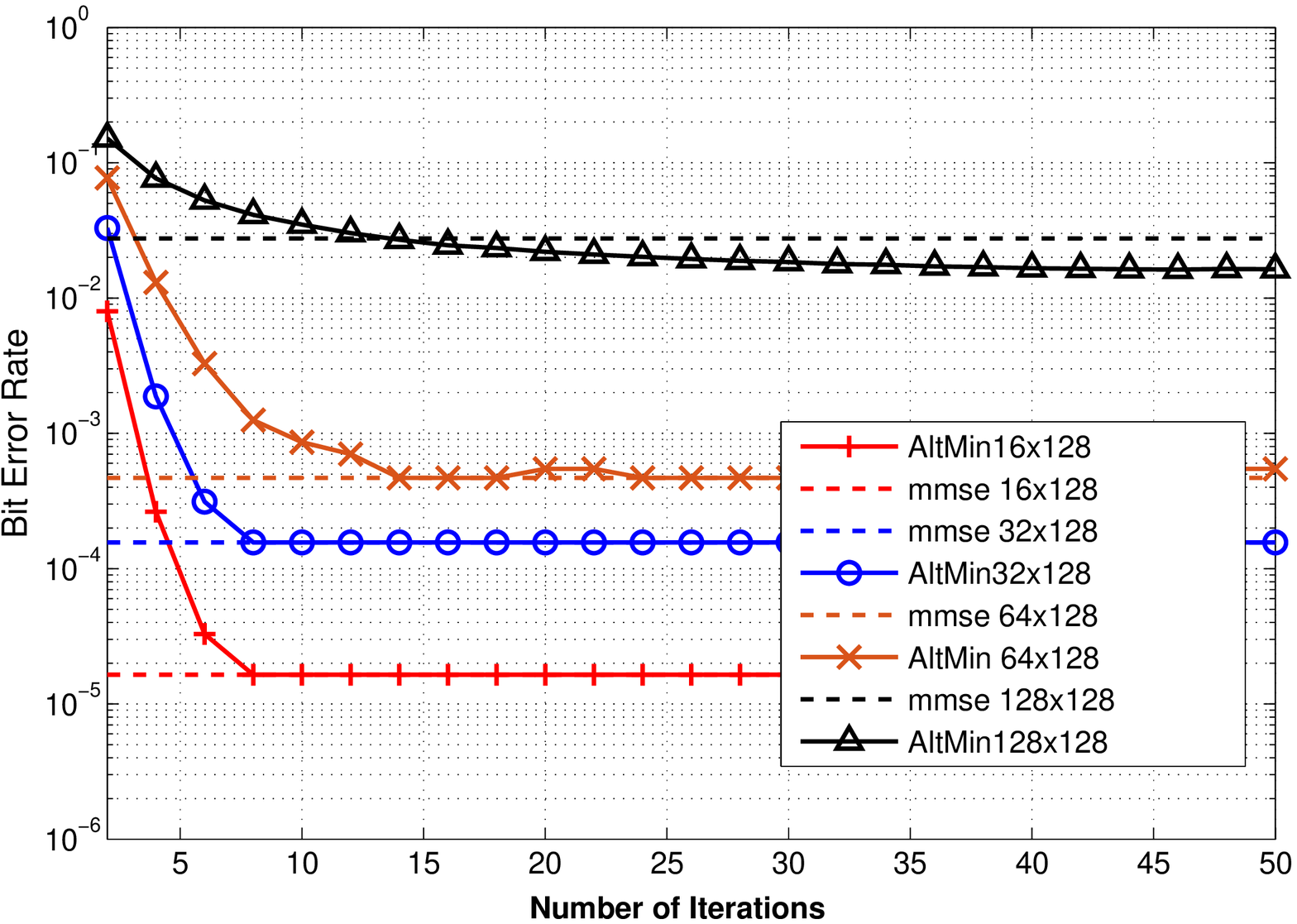}
  \caption{\small  BER performance versus AltMin Iterations at SNR=12 dB}
  \label{fig: 2}
\endminipage\hfill
\minipage{0.32\textwidth}
  \includegraphics[width=6cm, height=4cm]{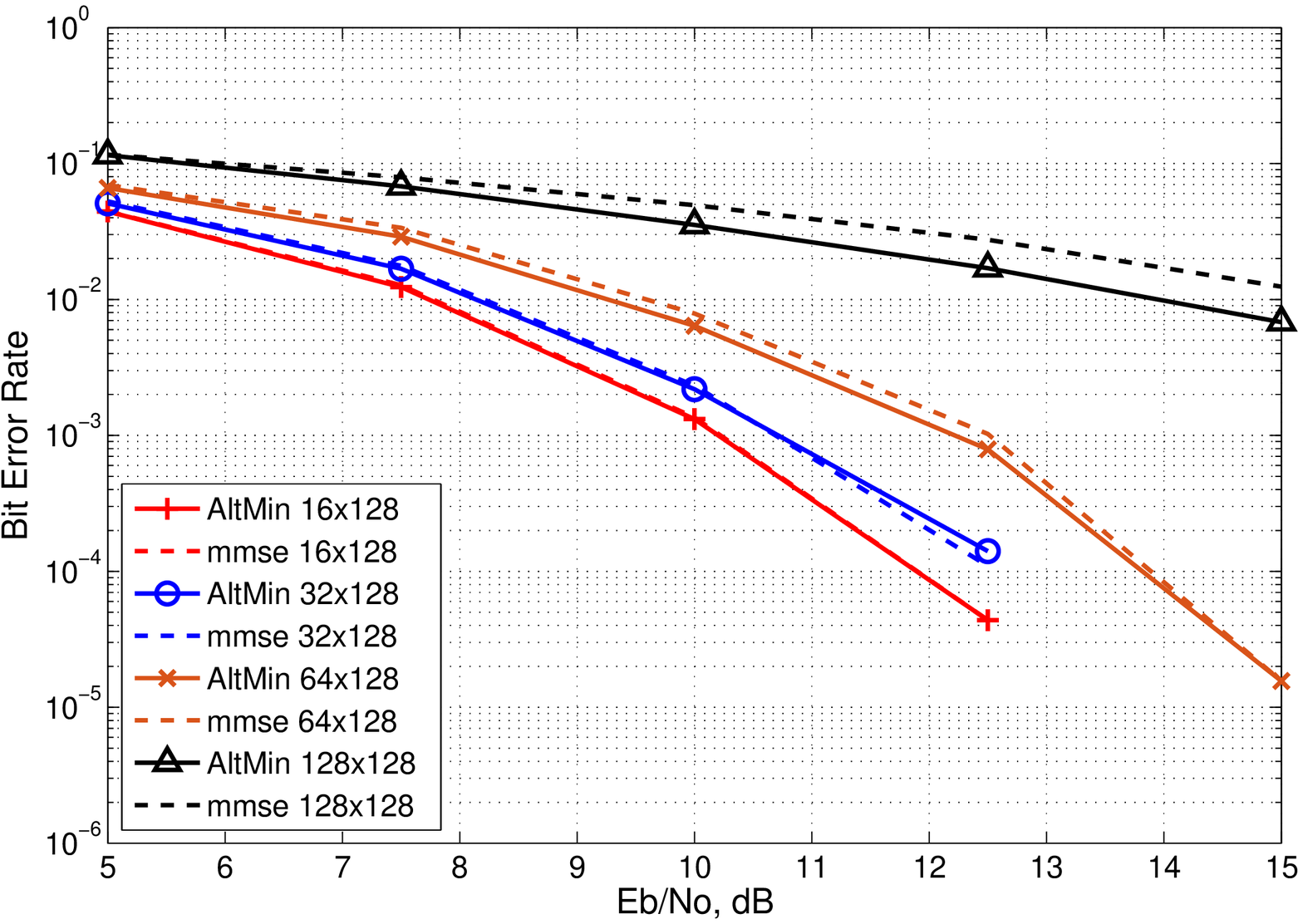}
  \caption{\small BER performance comparison for different massive MIMO configurations }
  \label{fig: 1}
\endminipage\hfill
\minipage{0.32\textwidth}%
  \includegraphics[width=6cm, height=4cm]{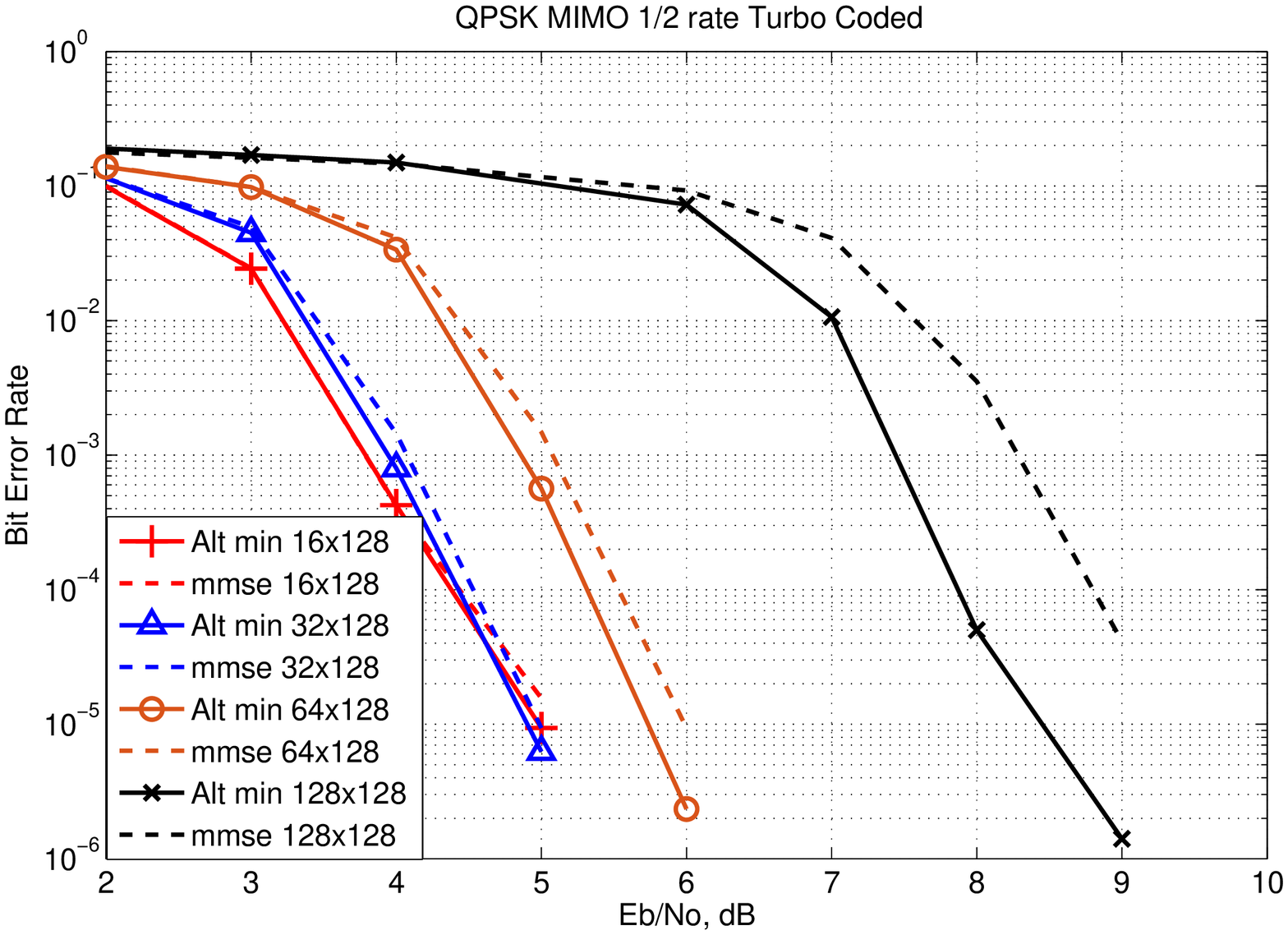}
  \caption{\small 1/2 Turbo coded BER performance of QPSK}
  \label{fig: 7.0}
\endminipage
\vspace{-.2in}
\end{figure*}

\vspace{-.1cm}
\subsection{Number of AltMin Iterations }
In this simulation experiment, we examine the number of iterations required by the AltMin algorithm such that the BER performance of both the proposed algorithm and the MMSE technique are equal, for various massive MIMO configurations. We fix convergence tolerance $\delta $ at $10^{-3}$, while the maximum number of iterations $T$ is  changing in a step size of 2. We also set the initial guess of $\textbf{x}$ in the AltMin algorithm to zeros, and $C=N_t$ in the update of $\lambda^{(k)}$ in (\ref{eq: 11}). 

Fig. \ref{fig: 2} shows BER performance versus maximum number of iterations for $N_t \times N_r =16\times 128, 32 \times 128, 64\times 128$, and $128 \times 128$ configurations, at SNR=12 dB. It can be noticed that the larger the ratio $\frac {N_r}{N_t}$ is, the smaller the number of iterations is required by AltMin to reach the MMSE performance. For example, 8 iterations is required for $16 \times 128$ configuration, while 15 iterations is required for $128 \times 128$ configuration. This results show that, on the average, to reach the MMSE performance, the required number of iterations, $T$, is much smaller than $N_t$. Consequently, the complexity of the proposed algorithm becomes in the order of $O( N_t N_r)$.%It can also be observed that no matter how as the number of iterations increases
%The MMSE performance is also shown in the same figure to examine the number of iterations required by AltMin to match the MMSE performance.
\if0
\begin{figure}[htp]
%%%%\begin{minipage}[b]{1.0\linewidth}
\centering
\includegraphics[width=8cm, height=6cm]{new_results/qpsk_effect_iters_BER_ntx128_snr12.eps}
\caption{\small  BER performance versus $\#$ of Iterations of AltMin for different massive MIMO configurations at SNR=12 dB }
\label{fig: 2}
\end{figure}
\fi

\vspace{-.1cm}
\subsection{BER performance comparison}
In this subsection, we present the BER performance of both the proposed algorithm and the MMSE technique with respect to SNR at various massive MIMO configuration. For each SNR value, we stop the AltMin algorithm at the iteration number at which its performance matches the MMSE performance based on the results from Fig. \ref{fig: 2}.% The number of BS antennas is 128, while the total number of users antennas can vary such that $N_t \leq N_r$.

%The proposed algorithm is compared to the MMSE linear detector at various massive MIMO configuration such as $16\times 128, 32 \times 128, 64\times 128, 128 \times 128$, that is the number of BS antennas is 128, while the number of single antenna users can vary such that $N_t \leq N_r$. For each SNR value, we stop the algorithm at around the iteration number at which its performance matches the MMSE performance based on the results from Fig. \ref{fig: 2}. 

Fig. \ref{fig: 1} shows that the BER of the proposed algorithm is upper bounded by that of the MMSE with the exact matrix inversion. For example, at higher ratio between $N_r$ and $N_t$, such as $16 \times 128$ or $32 \times 128$, the performance of the proposed algorithm and MMSE are the same. As the ratio becomes closer to 1, the BER of the proposed algorithm becomes slightly lower than that of the MMSE technique.% This advantage in the BER performance comes with fewer computations needed by AltMin, as shown in Table \ref{compComplexity}. In all the plotted performance in this figure, 

%shows the BER versus SNR performance results. It can be observed that at higher ration between $N_r$ and $N_t$, such as $16 \times 128$ and $32 \times 128$, the performance of AltMin and MMSE are the same. As the ratio becomes closer to 1 the performance of AltMin starts to exceed MMSE performance. For instance, at $10^{-2}$ BER, $128 \times 128$ AltMin with $40$ iterations has 1.2 dB superior than MMSE. This advantage in the BER performance comes with fewer computations needed by AltMin, as shown in Table \ref{compComplexity}. In all the plotted performance in this figure, 

 %The proposed algorithm is examined to see its performance when the  number of  antennas increases. This important aspect in large MIMO detectors referred to as adherence to a large system behavior. It means that the performance of a MIMO detector increases as the number of antennas increases \cite{elg015quad}. Fig. \ref{fig: 2} shows that the BER performance of the ...... improves as $N_t \times N_r$ increases (e.g. , $ 8\times 8$, $ 16\times 16$, $ 32\times 32$, and $ 64\times 64$). For instance at BER= $10^{-5}$, the performance of $ 64\times 64$ is just about 1 dB away from SISO AWGN.
\if0 
\begin{figure}[htp]
%%%%\begin{minipage}[b]{1.0\linewidth}
\centering
\includegraphics[width=8cm, height=6cm]{new_results/BER_vs_snr_MIMO_ntx128_alternate_algo.eps}
%%\includegraphics[width=8cm, height=6cm]{LASSO_2ADMM_qpsk_withRTS}
\caption{\small BER performance comparison with different massive MIMO configurations }
\label{fig: 1}
\end{figure}
\fi%

\vspace{-.1in}
\subsection{{Turbo Coded BER Performance}} The turbo coded BER performance of the proposed algorithm compared to the MMSE technique is shown in Fig. \ref{fig: 7.0} using coded QPSK modulation. In this simulation, all the above massive MIMO configurations are examined with rate-1/2 turbo encoder and decoder of 10 iterations. %$ \pm1 $ output valued vector from all detectors is fed as an input to the BCJR-based turbo decoder. 
In Fig. \ref{fig: 7.0},  AltMin based algorithm performs similar to the MMSE detector for $16\times 128$, and slightly better than MMSE for $32\times 128$. As the number of uplink antennas increases, the coded AltMin based algorithm outperforms coded MMSE, for example the improvement for the case of $128\times 128$ at $10^{-3}$ coded BER is about 1 dB compared to only 0.2 dB improvement in the case of $64\times 128$. 

%\indent Fig. \ref{fig:5} shows that the 2-LASSO-ADMM algorithm performs well even at higher QAM modulations, such as 16QAM, especially at high SNR regime. It outperforms QP detector at SNR greater than 17 dB, and RTS at SNR greater than 23 dB. Although RTS performs slightly better than our algorithm at low SNR with 16QAM, it was shown at various references that RTS tends to have a degraded BER performance as SNR increases and also as the number of antennas increases, especially at higher QAM modulations \cite{elghariani2016low}.

\vspace{-.1cm}
\subsection{Computational Complexity Analysis}
In this subsection, we compare the computational complexity of the proposed algorithm with the MMSE technique in terms of the number of multiplication operations, as depicted in Table~\ref{compComplexity}.  The comparison is based on the same SNR of $12$ dB and the same BER performance. More specifically, the number of iterations of AltMin is taken based on the results of Fig. \ref{fig: 2}, at which the BER of the two techniques coincides. 

%One thing to mention is that we found that scaling the $\lambda$ update equation, equation~\eqref{eq: 11} by $N_t$ minimize the number of iterations significantly. Therefore, we consider this scaling factor in updating $\lambda$.
%Moreover, to reduce the complexity of every iteration further, in the update of $\textbf{x}_i, \forall i$, we only consider solving equation~\eqref{eq:x1Update}. i.e, we don't consider the boundary solutions since we are already approximating the solution of the original non-convex problem and anyway we round the final values to $l$ and $-l$ if they are close to the boundary. In fact we did not notice any performance drop compared to considering the boundary solutions, but we run less number of computations. 
%First, we note that the AltMin needs to run for 8, 8, 14, and 14 iterations respectively for $16\times 128$, $32\times 128$, $64\times 128$, and $128\times 128$ in order to achieve the exact performance that is achieved by the MMSE. Table~\ref{compComplexity} shows the corresponding complexity of each of the algorithms in the total number of multiplication operations divided by $10^6$ at every MIMO configuration settings. 

From Table~\ref{compComplexity}, it can be observed that at small $N_t$, such as 16, MMSE outperforms the proposed algorithm by approximately a factor of 3. While for a large $N_t$, such as 128, the proposed algorithm shows superior computational reduction by a factor of $\frac{1}{6}$ as compared to MMSE. Note, although our algorithm requires more computations than MMSE for small $N_t$, it does not exhibit any matrix inversion or matrix multiplications, which is more advantageous in terms of hardware implementations \cite{zhang2018low}. 
%\cite{fang2016low}

%However, as the number of antennas increases, the computations needed by MMSE grows exponentially as compared to the linear growth of AltMin. Therefore, when $N_t$ reaches $64$ and $128$, AltMin archives a computation complexity that is approximately $2/3$-rd and $1/6$-th of the MMSE complexity respectively. Hence, we conclude that for large $N_t$, AltMin can achieve the performance of MMSE with significantly less computation complexity.
\small
\setlength\extrarowheight{3pt}

\begin{table}[h!]
%	  \vspace{-.1in}
  \centering
  \caption{Complexity comparison in terms of \# of real multiplications operations $\times 10^{6}$, for $N_r=128$}
  \begin{tabularx}{0.48\textwidth}{|X|X|X|X|X|} \hline
     & $N_t$=16 & $N_t$=32 & $N_t$=64 & $N_t$=128 \\ \hline
  MMSE  & 0.057 & 0.311 & 2.195 & 16.97 \\ \hline
  AltMin  & 0.204 & 0.409 & 1.409 & 2.818 \\ \hline
  $\frac{\text{AltMin}}{\text{MMSE}}$  & 3.57 & 1.31 & 0.64 & 0.166 \\ \hline
  \end{tabularx}
  \label{compComplexity}
  \vspace{-.1in}
\end{table}
\normalsize

\section{Conclusion}
In this letter we propose an iterative low complexity algorithm based on Alternating Minimization. This algorithm is a better alternative for the MMSE technique in Massive MIMO applications especially when the ratio between number of BS antennas and number of user equipment antennas (across all users) is small. We show that the proposed algorithm avoids complicated matrix inversion by solving the reformulated ML problem in an iterative manner, in which each iterations performs a simple computations based on a closed form expression. The results reveal that the algorithm can provide a lower computational complexity as compared to the MMSE technique with exact matrix inversion for both coded and uncoded cases.
%n this letter, we proposed an iterative low complexity algorithm based on Alternating Minimization approach. This algorithm is a good alternative for the MMSE technique in Massive MIMO applications where the number of BS antennas is not very large compared to the number of user equipment antennas.  The proposed Algorithm re-formulates the ML detection problem as a sum of convex functions based on decomposing the received vector into multiple vectors. Each vector represents the contribution of one of the transmitted symbols in the received vector. Then AltMin is used to solve the new formulated problem in an iterative manner with a closed form solution update in every iteration. The results reveal that the algorithm can provide lower computational complexity as compared to the exact MMSE with similar BER performance in both coded and uncoded performance.

\bibliographystyle{IEEEtran}

\bibliography{refMIMO}

\end{document}